\documentclass[12pt]{article}
\usepackage[T1]{fontenc}
 \usepackage{kpfonts}
\usepackage[dvipsnames]{xcolor}
\usepackage{pdfpages}
\usepackage{sgame}
\usepackage{accents}
\usepackage{tikz-cd}
\usepackage{float}
\usepackage[round]{natbib}   
\usepackage{multirow}
\usepackage{dutchcal}
\usepackage{pdfpages}
\usepackage{sgame}
\usepackage{mathtools}
\usepackage{amsthm}
\usepackage{enumitem}
\usepackage[margin=1in]{geometry}
\usepackage{epigraph}
\usetikzlibrary{calc}
\usetikzlibrary{arrows}

\newcommand{\ubar}[1]{\underaccent{\bar}{#1}}

\newtheorem{theorem}{Theorem}[section]

\newtheorem{lemma}[theorem]{Lemma}

\theoremstyle{definition}

\newtheorem{example}[theorem]{Example}

\newtheorem{remark}[theorem]{Remark}

\usepackage{hyperref}
\hypersetup{
    colorlinks=true,
    linkcolor=Violet,
    filecolor=Violet,      
    urlcolor=Violet,
    citecolor=Violet,
}

\providecommand{\keywords}[1]{\textbf{\textit{Keywords:}} #1}
\providecommand{\jel}[1]{\textbf{\textit{JEL Classifications:}} #1}

\begin{document}
\author{Mark Whitmeyer\thanks{Hausdorff Center for Mathematics \& Institute for Microeconomics, University of Bonn \newline Email: \href{mailto:mark.whitmeyer@gmail.com}{mark.whitmeyer@gmail.com}. \newline I am grateful to an anonymous referee for a different paper for the comment that sparked this idea. I also thank \"Ozlem Bedre-Defolie, Vasudha Jain, Stephan Lauermann, Benny Moldovanu, Joseph Whitmeyer, and Thomas Wiseman for helpful comments and feedback. This work was generously funded by the DFG under Germany's Excellence Strategy-GZ 2047/1, Projekt-ID 390685813.} }

\title{Persuasion Produces the (Diamond) Paradox}


\date{\today}

\maketitle

\begin{abstract} This paper extends the sequential search model of \cite{asher} by allowing firms to choose how much match value information to disclose to visiting consumers. This restores the Diamond paradox (\cite{diamond}): there exist no symmetric equilibria in which consumers engage in active search, so consumers obtain zero surplus and firms obtain monopoly profits. Modifying the scenario to one in which prices are advertised, we discover that the no-active-search result persists, although the resulting symmetric equilibria are ones in which firms price at marginal cost.
\end{abstract}
\\
\jel{C72; D82; D83}
\newpage

\setlength{\epigraphwidth}{2.5in} \epigraph{I'm gonna stop, uh, shoppin' around.}{Elvis Presley (\href{https://www.youtube.com/watch?v=ADjm8yzYFW4}{G.I. Blues})}

\section{Introduction}

The Diamond paradox (\cite{diamond}) is a stark result that highlights the importance of search frictions in models of price competition. Infamously, Diamond establishes that even in a market with a large number of firms, an arbitrarily small (yet positive) search cost ensures that the firms behave like monopolists and that consumers do not search. The intuition behind this result is well-known: for any price strictly below the monopoly price, demand is locally perfectly inelastic and so a firm can always improve its lot by raising its price slightly. 

As a number of subsequent papers illustrate, small modifications to the model can overturn the result. Two such works stand out in particular: \cite{sta} assumes that a fraction of the market's consumers have zero search cost and therefore freely buy from whomever sets the lowest price. As a result, demand for each firm slopes down again and they behave like monopolists no longer. \cite{asher} takes a different approach. In his model, although each consumer faces a positive search cost, the firms' products are differentiated and consumers have imperfect information about the products. Each consumer's match value at any firm is an i.i.d. random variable, about which the firms and consumer are \textit{ex-ante} uninformed. Upon visiting a firm, a consumer discovers both the firm's price as well as the realized match value. This uncertainty begets equilibria with active search. 

In this paper, we revisit the framework of \cite{asher} but alter it by assuming that each firm may choose how much information to provide to a consumer during her visit, in addition to setting a price. Each firm commits to a signal that maps a consumer's match value to a (conditional) distribution over signal realizations. In the spirit of the original Diamond paper, neither the signal nor the price chosen by a firm is observable until a consumer incurs the search cost. That is, after paying the search cost to visit a firm, a consumer then observes that firm's signal, price, and draws a signal realization that yields her posterior belief about the firm's match value. 

This modification has drastic consequences. We find that even in a model with product differentiation and imperfect information, information provision \textit{restores the Diamond paradox}. Namely, the main result of this paper, Theorem \ref{main}, states that there are no symmetric equilibria with active consumer search. Each consumer visits at most\footnote{There is a trivial equilibrium in which consumers conjecture extremely high prices and do not visit any firms. As is convention, throughout this paper, we ignore this equilibrium.} one firm and purchases from it. In all symmetric equilibria, firms leave consumers with zero rents.

The intuition to this result is similar to that in the original paper by Diamond. There, in any purported equilibrium with active search, firms can always exploit visiting consumers by raising their prices slightly, which does not affect their purchase decisions due to the positive search cost. Here, firms deviate by providing slightly less information and pooling beliefs above consumers' stopping thresholds. Subsequently, given that any equilibrium must involve no search, the classic Diamond paradox incentive kicks in and firms must obtain monopoly profits. 

This incentive to pool beliefs is extremely strong, and continues to drive the results even when prices are posted and so can be observed before consumers embark on their searches. The second finding of this paper, Theorem \ref{main2}, is that even when prices are posted--and therefore shape consumers' search behavior directly--there exist no symmetric equilibria with active search. Again, any purported equilibrium in which there is active search would allow firms to deviate profitably by providing less information. In the unique symmetric equilibrium, because prices are posted, the usual Bertrand forces apply such that firms price at marginal cost and obtain no profits. Consumer surplus is merely the expected match value of the firm and no useful information is provided. 

This paper belongs to the growing collection of papers that explore information design and persuasion in consumer search settings. This literature includes \cite{board}, who also provide conditions for an analog of the Diamond paradox to arise. Their model is completely different, however: the uncertain state of the world is common and so the competing sellers each provide information about the common state to prospective consumers. If sellers can perfectly observe a consumer's current belief about the state upon her visit, or coordinate their persuasion strategies, then there is an equilibrium in which they provide the monopoly level of information and set the monopoly price. This equilibrium is not generally unique and requires (mild) additional assumptions that pertain to the competitive persuasion problem with a common state. Furthermore, the sellers in \cite{board} do not set prices but compete through information alone.

In contrast, this paper is a direct adaptation of \cite{asher}, in which the distributions over match values are endogenous objects. Importantly, a consumer's match value at any firm is i.i.d., so a firm's choice of distribution does not affect a consumer's belief at any other firm. This paper shows that the paradox occurs when firms set prices as well, and even holds when prices are posted (advertised) and so search is directed. Due to the independence of match values, firms know a visiting consumer's prior about their quality, but not her outside option, which is endogenously determined by her search behavior. Thus, although a firm may not know precisely where it is in a consumer's search order, it knows the upper bound of a consumer's outside option and can tailor its information accordingly.

Two other papers similar to this one are \cite{perc} and \cite{whau}. The first explores a game of pure information provision (without prices) between sellers in \cite{wei}'s classic sequential search setting. There, if consumers must incur a search cost to discover a firm's match value distribution, firms do not provide any (useful) information and consumers do not search actively, as in this paper. \cite{whau}, in turn, modifies \cite{perc} by allowing firms to compete both through prices and (advertised) information. In some sense there is an asymmetry--with posted information, there exist symmetric equilibria with active search even when prices are hidden, but as we discover here, the mirrored statement is false.

Another relevant paper is \cite{hu}, who look at consumer-optimal information structures in Wolinsky's model of sequential search. Related to that is recent work by \cite{zhou}, who looks broadly at how improved information affects consumer welfare in Wolinsky's setting. \cite{za}, in turn, focus on consumer-optimal information structures in a (frictionless) duopoly market. Like this paper, \cite{hkb} also look at competitive information design and price setting in an oligopoly market, albeit one without search frictions. In their framework, absent such frictions, they show that firms provide (some) information to consumers and that both consumers and firms obtain some surplus in the market.

Finally, this paper relates naturally to the vein of research that focuses on obfuscation. One particularly relevant paper is \cite{wol}, who modify the model of \cite{sta} by allowing firms to choose the length of time necessary for consumers to learn their prices. Like firms' information policies in this paper, the required time to learn the price at a firm is hidden from consumers until their visits. A pair of important distinctions between their paper and this one; however, are that products in their model are undifferentiated and the level of obfuscation chosen by a firm affects consumers by altering their \textit{future} search costs. 

In contrast, here, firms choose how much information about their products consumers acquire during their visits, and therefore shape consumers' valuations directly. Thus, any alteration by a firm to its information policy may come at a cost, since it could potential limit the rents the firm could extract. Moreover, prices in their model are hidden (given the market's homogeneity, a posted price version of their model is just the canonical Bertrand setup), whereas we find that the same economic intuition drives this paper's results regardless of whether prices are posted. Firms could (costlessly) provide useful information, yet do not do so in any symmetric equilibrium.

\section{Model}

The foundation for this paper is the workhorse sequential search model of \cite{asher}. There are $n$ symmetric firms and a unit mass of consumers with unit demand. The match value of a consumer at firm $i$ is an i.i.d. random variable, $X_{i}$, distributed according to (Borel) cdf $G$ on $\left[0,1\right]$. Let $\mu$ denote the expectation of each $X_{i}$. Given a realized match value of $x_{i}$ and price $p_{i}$, a consumer's utility from purchasing from firm $i$ is 
\[u\left(x_{i},p_{i}\right) = x_{i} - p_{i} \text{ .}\]

In contrast to the original model of \cite{asher}, a consumer does not directly observe her match value at firm $i$, but instead observes a signal realization that is correlated with it. Each firm has a compact metric space of signal realizations, $S$, and commits to a signal, Borel measurable function $\pi_{i}: \left[0,1\right] \to \Delta(S)$. As is well known, each signal realization begets--via Bayes' law--a posterior distribution over values, and thus a signal begets a (Bayes-plausible) distribution over posteriors. Alternatively, a signal begets a distribution over posterior means, and the following remark is now standard in the literature:

\begin{remark}
Each firm's choice of signal, $\pi_{i}$, is equivalent to a choice of distribution $F_{i} \in \mathcal{M}(G)$, where $\mathcal{M}\left(G\right)$ is the set of all mean-preserving contractions of $G$.
\end{remark}

Thus, each firm chooses a distribution over values, $F_{i}$, and sets a price $p_{i}$. Importantly, a consumer only observes these choices upon her visit to a firm. Following \cite{asher}, search is sequential and with recall. At a cost of $c > 0$, a consumer may visit a firm and observe its price and realized draw from distribution $F_i$. As does \cite{diamond}, we assume that a consumer incurs no search cost for her first visit. For simplicity, we impose that the marginal cost for each firm is $0$.

We look for symmetric equilibria in which each firm sets the same price $p$ and chooses the same distribution over values $F$.\footnote{\cite{asher} focuses on symmetric pure strategy pricing equilibria, which emphasis is echoed in this paper. Nevertheless, the results of this paper hold for all symmetric equilibria, in both pure and mixed strategies.} Our equilibrium concept is (weak) Perfect Bayesian Equilibrium (PBE), in which consumers have the same beliefs about the (theretofore unobserved) prices and information policies (signals) chosen by unvisited firms on and off the equilibrium path. This is customary in the  literature and seems reasonable.\footnote{This assumption; however, is not innocuous. See \cite{janssen2020beliefs} for an in depth exploration of the ramifications of this assumption. As they note there--and in an earlier paper, \cite{janssen2015consumer}--this stipulation may be more objectionable in markets with vertical relations.} For simplicity, a consumer has an outside option of $0$. 

It is a standard result that a consumer follows a reservation price policy. Indeed, a consumer's search problem is a special case of the one explored in \cite{wei}. Define the reservation value, $z$, induced by the conjectured price $\tilde{p}$ and distribution $\tilde{F}$ by
\[\label{eq1}\tag{$\star$} c = \int_{z+\tilde{p}}^{1}\left(x - \tilde{p} - z\right)d\tilde{F}(x) \text{ .}\]
In a symmetric equilibrium, a consumer's optimal search protocol is to visit firms in random order and stop and purchase from firm $i$ if and only if the realized value at that firm, $x_{i}$, satisfies $x_{i} - p \geq z$, where $p$ is the actual price set by firm $i$. If $x_{i} < z + p$ for all $i$ then a consumer selects the firm whose realized value $x_{i}$ is highest. At equilibrium, the conjectured price, $\tilde{p}$, must equal the actual price set by each firm, $p$; and the conjectured distribution, $\tilde{F}$, must equal the actual distribution chosen by each firm, $F$.

Note that we are using \cite{wei}'s formulation of a consumer's stopping problem, which is slightly nonstandard but equivalent to that of \cite{asher}. This allows for an easy transition into the next subsection, wherein we allow for advertised (posted) prices.

Now, let us establish the main result. The first step is to derive the following lemma.
\begin{lemma}\label{nosym}
There are no symmetric equilibria in which consumers visit more than one firm.
\end{lemma}
\begin{proof}
Suppose for the sake of contradiction that there is such an equilibrium. For expositional ease, let us begin by assuming that the firms' choices of $F$ and $p$ are deterministic. Because $c > 0$ and by the definition of the reservation value, the conjectured distribution, $\tilde{F}$, must be such that both values strictly below and strictly above $z+ \tilde{p}$ occur with strictly positive probability. Accordingly, let $\left[\ubar{a},\bar{a}\right]$ and $\left[\ubar{b},\bar{b}\right]$ be intervals such that $\int_{\ubar{b}}^{\bar{b}}d\tilde{F}(x) > 0$ and $\int_{\ubar{a}}^{\bar{a}}d\tilde{F}(x)> 0$, where $\ubar{a} \leq \bar{a} < z + \tilde{p}$ and $z + \tilde{p} < \ubar{b} \leq \bar{b}$.

Given price $\tilde{p} > 0$,\footnote{Obviously, since $c > 0$, there are no symmetric equilibria in which the market price is $0$.} a firm's payoff from any value, $x$, that is weakly greater than $z + \tilde{p}$ is $\tilde{p}$. Moreover, its payoff from any value, $x$, that is strictly lower than $z + \tilde{p}$ is strictly less than $\tilde{p}$. Let $\alpha < \tilde{p}$ denote its average payoff (under $\tilde{F}$) for values in the interval $\left[\ubar{a}, \bar{a}\right]$.

It is easy to see then that a firm can deviate profitably by choosing distribution $\hat{F}$, where $\hat{F}$ is constructed from $\tilde{F}$ by taking the measure on $\left[\ubar{b},\bar{b}\right]$ and some fraction $\epsilon > 0$ of the measure on $\left[\ubar{a},\bar{a}\right]$ and collapsing them to their barycenter, and is set equal to $\tilde{F}$ everywhere else.\footnote{In the parlance of \cite{hill}, $\hat{F}$ is a fusion of distribution $\tilde{F}$. The notions of fusion and mean-preserving-contraction are (in this paper) equivalent.}

$\hat{F}$ will have a point mass on some $\hat{x}$, which will occur with probability $\int_{\ubar{b}}^{\bar{b}}d\tilde{F}(x) + \epsilon \int_{\ubar{a}}^{\bar{a}}d\tilde{F}(x)$. For $\epsilon$ sufficiently small, $\hat{x} > z + \tilde{p}$, and by construction $\hat{F} \in \mathcal{M}\left(\tilde{F}\right)$. Thus, since $\tilde{F} \in \mathcal{M}\left(G\right)$, $\hat{F} \in \mathcal{M}\left(G\right)$. Finally, the net change in the firm's payoff is $\epsilon \int_{\ubar{a}}^{\bar{a}}d\tilde{F}(x) \left(\tilde{p}-\alpha\right) > 0$, which concludes the proof.

It is simple to modify the proof to accommodate mixing by firms. In such a (purported) symmetric equilibrium, each firm chooses the joint distribution $F(x, p)$, over values and prices, $\left[\ubar{x},\bar{x}\right] \times \left[\ubar{p},\bar{p}\right]$, where each conditional distribution over values, $F_{X|P}\left(x|P = p\right) \in \mathcal{M}\left(G\right)$. No randomness is resolved, however, until consumers visit firms and so we may define a new random variable $Y \coloneqq X - P$ with distribution $H$ on some interval.\footnote{Clearly, because prices are non-negative, the upper bound of the support must be (weakly) less than $1$.} 

Accordingly, the reservation value induced by a consumer's conjectured $\tilde{H}$ is 
\[c = \int_{z}^{1}\left(y-z\right)d\tilde{H}(y) \text{ .}\]
Note that the induced reservation values, $z$, and distributions, $H$, chosen by firms are deterministic, and identical (because this a symmetric equilibrium). 

The remainder is (more-or-less) identical to the pure strategy case with some minor subtleties. Because consumers are searching actively, $\tilde{H}$ must be such that a consumer strictly prefers to stop and strictly prefers to continue with strictly positive probability. There are two cases: first, it is possible that for some $p'$ in the support of a firm's mixed strategy, the associated $\tilde{F}_{X|P}\left(x|p'\right)$ has the same property as $\tilde{H}$--namely, both (strict) stopping and (strict) continuation values occur with strictly positive probability. As in the pure strategy case, a firm can deviate profitably by pooling these values carefully.

On the other hand, it may be possible that there is no such $p'$. However, then, there must exist some $p''$ and $\tilde{F}_{X|P}\left(x|p''\right)$ in the support of a firm's mixed strategy such that, given $p''$, stopping values occur with probability one and strict stopping values realize with strictly positive probability. Consequently, a firm can deviate by providing no information (choosing the distribution $\delta_{X}\left(\mu\right)$)\footnote{$\delta_{X}\left(\mu\right)$ denotes the degenerate distribution $\mathbb{P}\left(X = \mu\right) = 1$.}--which, at price $p''$, must leave the consumer with strict incentive to stop--and instead charging some price $p''+\epsilon$, for a sufficiently small (yet strictly positive) $\epsilon$.
\end{proof}
In the standard Diamond paradox, firms have an incentive to raise prices slightly in order to take advantage of the quasi-monopoly power inferred upon them by the search cost. A similar effect occurs with regard to their information provision policies: they have an overwhelming incentive to provide slightly less information to consumers in order to increase the probability that they do not continue their searches.

\begin{example}\label{ex1}
Suppose there are infinitely many firms and that the conjectured distribution over match values is the uniform distribution on $[0,1]$, and the search cost is $c = 1/32$. Given the conjectured distribution the conjectured price must be $\sqrt{2c} = 1/4$--indeed, if this distribution over match values were exogenous and firms competed on price alone, this would be the unique symmetric equilibrium price (outside of the trivial, no-visit equilibrium). Moreover, the associated reservation value $z = 1/2$ and so a firm's profit is $1/16$. 

Now let a firm deviate in the manner proposed by the proof of Lemma \ref{nosym}. In particular, suppose it deviates by charging price $1/4$ and choosing a distribution that has a linear portion with slope $1$ on the interval $\left[0,1/2\right]$ and a point mass of size $1/2$ on value $3/4$. This brings the deviating firm an improved profit of $1/8$. The conjectured distribution and the described profitable deviation are depicted in Figure \ref{fig1}.

\begin{figure}
    \centering
    \includegraphics[scale=.15]{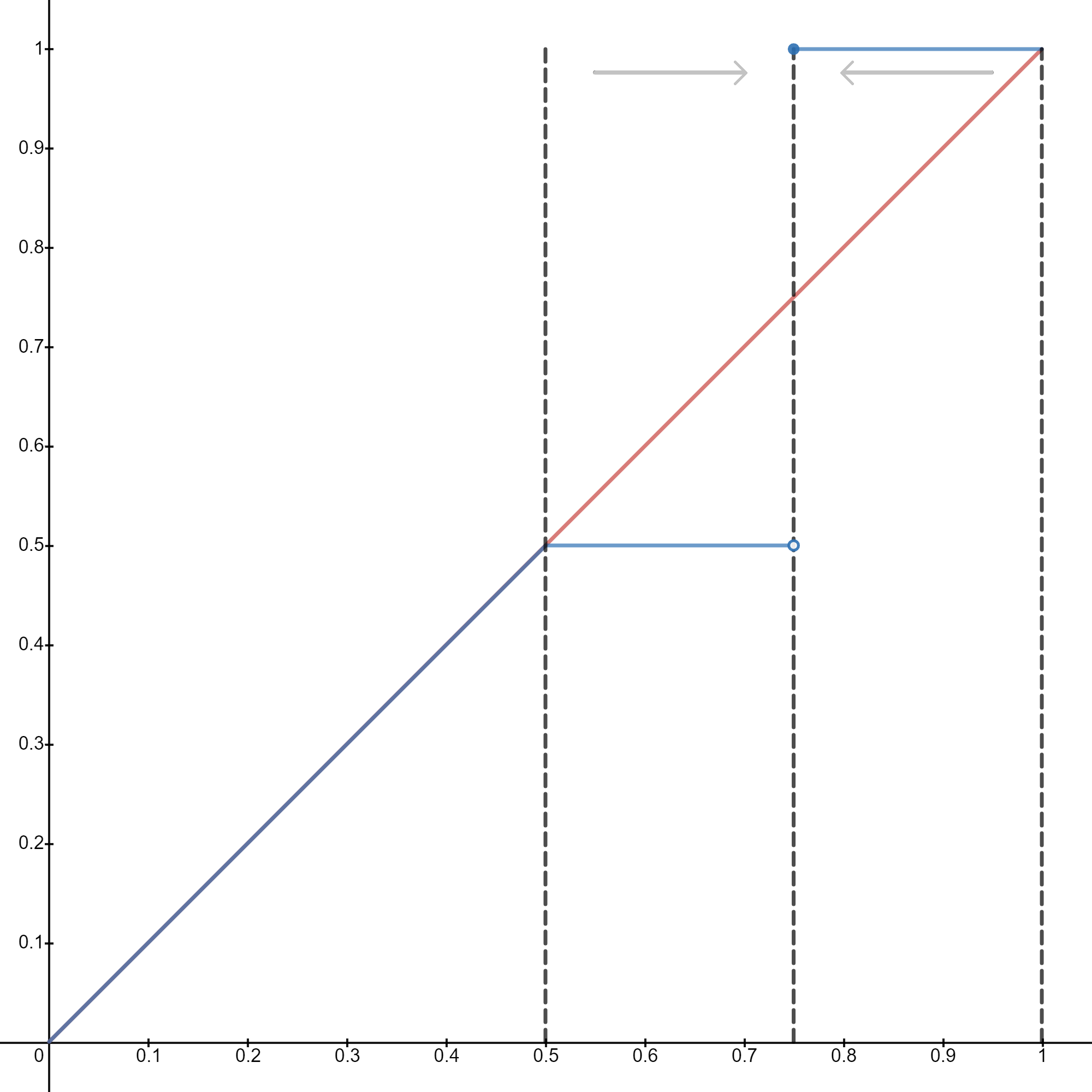}
    \caption{$\tilde{F}(x) = x$ (in red) and a profitable deviation, $\hat{F}(x)$ (in blue), from Example \ref{ex1}. The gray arrows indicate the fusing of measures (or the contraction of the distribution).}
    \label{fig1}
\end{figure}
\end{example}


Next, we find that, given that no consumer visits more than one firm, there are no equilibria in which consumers obtain any rents. 
\begin{lemma}
There are no symmetric equilibria in which consumers obtain any surplus.
\end{lemma}
\begin{proof}
From the previous lemma, we know that there are no symmetric equilibria in which consumers visit more than one firm. First, suppose that there is an equilibrium in which the reservation value chosen by each firm, $z$, is greater than $0$. Evidently, since a consumer visits just one firm, the reservation value must equal $\mu - c - \tilde{p}$, which is greater than $0$ by assumption. 

A firm's profit, conditional on a consumer's visit, is $\tilde{p} \leq \mu - c$. Clearly, then, a firm can deviate by charging price $p = \tilde{p} + c - \epsilon$. For $\epsilon > 0$ sufficiently small, this is a profitable deviation. Thus, there are no equilibria in which firms induce non-negative reservation values.

Second, suppose that there is an equilibrium in which the reservation value chosen by each firm, $z$, is strictly less than $0$. Now, each firm is facing a monopolist's problem where a consumer has an outside option of $0$. Accordingly, a firm can extract all of a consumer's surplus by providing no information (choosing the degenerate distribution $\delta_{X}\left(\mu\right)$) and charging $p = \mu$. 
\end{proof}

Combining the two lemmata, we have the first main result.
\begin{theorem}\label{main}
In any symmetric equilibrium in which firms are visited, each firm makes a sale with certainty, conditional on being visited, and obtains the monopoly profit of $\mu$. There is no active search and consumer surplus equals $0$.
\end{theorem}

\subsection{Advertised Prices}

Perhaps surprisingly, the no active search result continues to hold, \textit{even if prices are advertised} and so search is directed. Specifically, we assume that a consumer observes the price set by each firm before starting her search, but still must visit a firm to observe its (expected) match value. As noted by \cite{az}, ``price comparison websites are now a major part of the retailing website.'' In many markets, consumers can obtain price quotes for free, before embarking on their costly searches. Information, on the other hand, is perhaps more difficult to advertise, so it is plausible that firms can advertise prices but not how much information they provide. We find that even when prices are posted, there are no symmetric equilibria with active search. 

We continue to use the weak PBE solution concept. Since price deviations are now observable, we allow consumers the freedom to have any belief about a firm's signal following an off-path price choice by that firm. In the same spirit of the assumption in the hidden prices case, we stipulate that a price deviation by a firm does not affect consumers' conjectures about the distributions chosen by firms that choose on-path prices. Similarly, during a consumer's search, her beliefs about the signals chosen by theretofore unvisited firms are unaffected by other firms' deviations.

\begin{lemma}\label{ppone}
In any symmetric equilibrium in which firms make strictly positive profits, there is no active search. That is, consumers visit no more than one firm.
\end{lemma}
\begin{proof}
In any symmetric equilibrium, each firm chooses a joint distribution, $F(x, p)$, over values and prices, $\left[\ubar{x},\bar{x}\right] \times \left[\ubar{p},\bar{p}\right]$, where each conditional distribution over values, $F_{X|P}\left(x|P = p\right) \in \mathcal{M}\left(G\right)$. Given a firm's choice of (on-path) price $p^{*} \in \left[\ubar{p},\bar{p}\right]$, its (on-path) choice of distribution over values $F_{X|P}\left(x|P = p^{*}\right)$ must yield a maximal payoff for the firm, given the strategies of the other firms.

For each price, the corresponding conditional distribution over values corresponds to a reservation value, $z$, which is pinned down by Equation \ref{eq1}. Define $\mathcal{Z}$ to be the set of all reservation values that are induced on-path. Joint distribution $F(x, p)$ induces a distribution over reservation values $\Phi(z)$. Note that, in contrast to the previous section's case, in which both prices and distributions are hidden, the observability of firms' price choices means that the reservation values chosen by firms may be random.

Since firms are making strictly positive profits, any on-path price, $p^{*}$, must itself be strictly positive. Denote by $z^{*} \in \mathcal{Z}$ the corresponding reservation value. Evidently, a firm's payoff, conditional on being visited, from any value $x \geq \max\left\{0,z^{*}\right\} + p^{*}$, is $p^{*}$. 

First, suppose that there exists an on-path price $p^{*}$ and conjectured distribution $\tilde{F}_{X|P}\left(x|p^{*}\right)$ that induces a $z^{*}$ that is weakly greater than $0$. As in the proof for Lemma \ref{nosym}, suppose for the sake of contradiction that conditional on her arrival at the firm a consumer strictly prefers to stop and strictly prefers to continue her search (or select her outside option) with strictly positive probability. However, all values in the stopping set yield a payoff of $p^{*}$ and all beliefs in the continuation set yield a payoff that is strictly below $p^{*}$. Consequently, as in the proof for Lemma \ref{nosym}, after posting price $p^{*}$, a firm can (secretly) deviate by providing slightly less information and fusing a subset of the beliefs in the continuation and stopping sets.

Second, suppose that there exists an on-path $z^{*}$ that is strictly less than $0$. In such an equilibrium, if a firm is visited, it is the first and only firm that is visited by the consumer, since the first visit is the only one that does not impose on the consumer a search cost. 
\end{proof}

The reason why there may exist an on-path $z^{*}$ that is strictly less than $0$ is because consumers do not incur search costs at the first firm they visit. If they did incur such costs at the first firm, then all on-path $z^{*}$ would have to be weakly positive, since otherwise consumers would strictly prefer their outside option to visiting firms with such $z$ values.

\begin{example}\label{exa2}
Let $n = 2$ and the match value of each firm be a Bernoulli random variable with mean $1/2$. Let the search cost, $c = 1/16$. If firms could not choose how much information to provide--and were forced to provide full information--and could only compete by posting prices, then it is straightforward to see that there is an equilibrium in which each firm chooses a continuous distribution over prices that induces the distribution over reservation values, $\Phi\left(z\right)$, given by
\[\label{e1}\tag{$1$}\Phi\left(z\right) = \frac{7}{7-8z}-1, \quad \text{on} \quad \left[0,\frac{7}{16}\right] \text{ .}\]
The associated range of prices is $\left[7/16, 7/8\right]$. Now, let us maintain the assumption that prices are posted but restore the ability of firms to choose the distribution over match values. We will establish that given the conjectured distribution over reservation values stated in Expression \ref{e1}, there are on-path prices after which firms strictly prefer to deviate and choose a distribution other than the Bernoulli distribution. 

Indeed, let us characterize the optimal distribution over values following a choice of price $p^{*} = 7/16$ by a firm. Given a consumer's conjectured $\Phi$, because the firm has chosen the lowest on-path price, it will be visited first. This makes it easy to write the firm's payoff as a function of its posterior value, $V(x)$. It is
\[V(x) = \begin{cases}
0, \qquad &0 \leq x < \frac{7}{16}\\
\left(\frac{7}{21-16x}\right)\left(\frac{7}{16}\right), \qquad &\frac{7}{16} \leq x \leq \frac{7}{8}\\
\frac{7}{16}, \qquad &\frac{7}{8} \leq x \leq 1\\
\end{cases} \text{ .}\]
From \cite{kam}, we know that the optimal distribution over values can be obtained by concavifying $V$. $V$ and its concavification, $\hat{V}$, are depicted in Figure \ref{fig2}, in which the distribution over posteriors corresponding to full information (the Bernoulli distribution over values) is also portrayed. Evidently, the latter is strictly suboptimal.

\begin{figure}
    \centering
    \includegraphics[scale=.12]{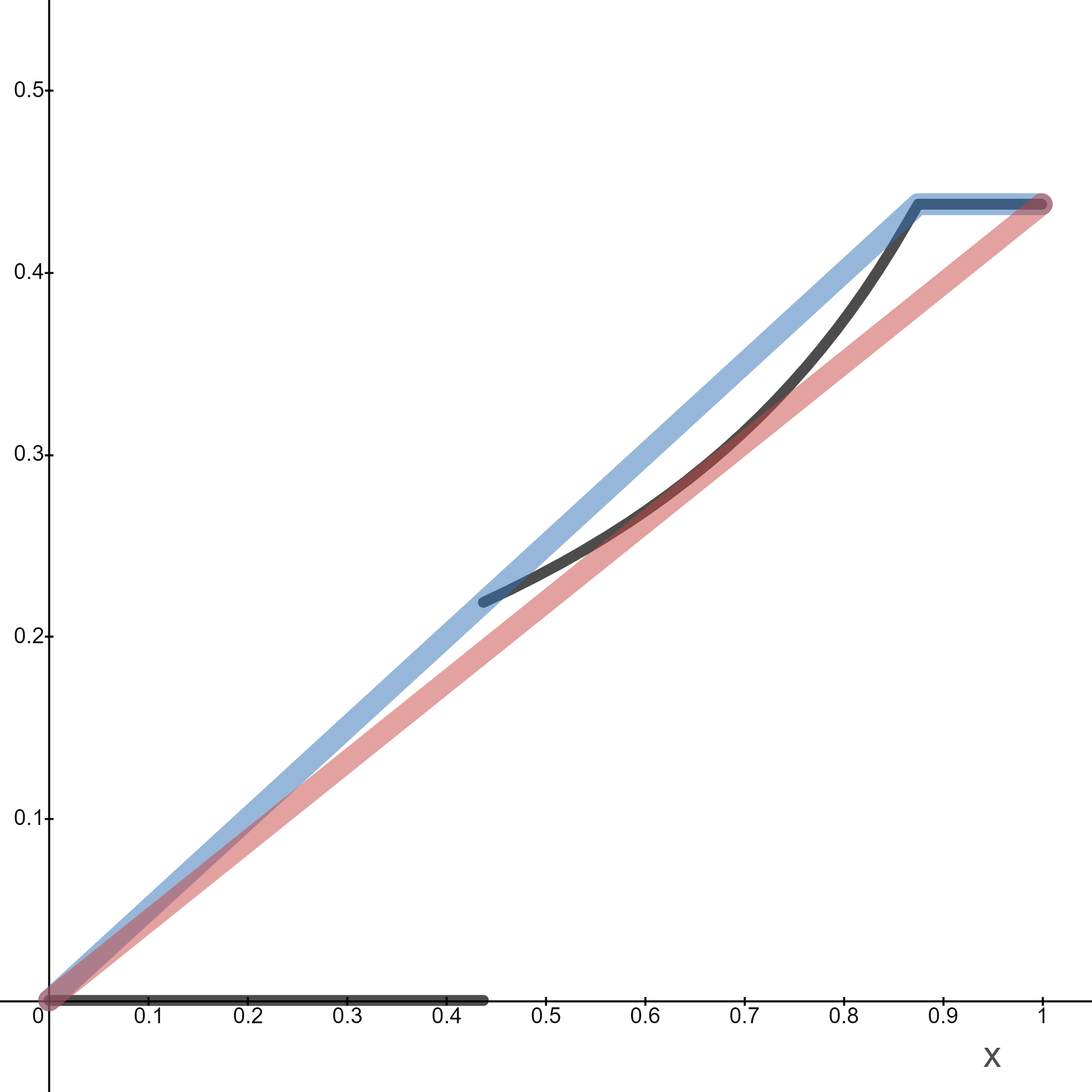}
    \caption{$V(x)$ (in black); its concavification, $\hat{V}\left(x\right)$ (in blue); and the splitting corresponding to the Bernoulli distribution (in red); from Example \ref{exa2}.}
    \label{fig2}
\end{figure}
\end{example}

The next lemma restricts what can happen at equilibrium following firms' price choices.

\begin{lemma}\label{aux}
In any equilibrium in which firms make strictly positive profits we must have the following:
\begin{enumerate}
    \item If $1 \geq p^{*} > \mu$ is chosen on-path, distribution $F_{X|P}\left(x|p^{*}\right)$ has a mass point of size $q^{*} > 0$ on $p^{*}$ and has no support strictly above $p^{*}$; 
    \item If $\mu \geq p^{*} > 0$ is chosen on-path, distribution $F_{X|P}\left(x|p^{*}\right)$ must have support entirely above $p^{*}$. A firm is selected with certainty, conditional on being visited.
\end{enumerate}
\end{lemma}
\begin{proof}
If $1 \geq p^{*} > \mu$ is chosen on-path, then if the associated reservation value, $z^{*} \geq 0$, the conjectured distribution must be such that $\int_{\ubar{x}}^{\bar{x}}d\tilde{F}_{X|P}\left(x|p^{*}\right) > 0$ for some $\Bar{x} \geq \ubar{x} > z^{*} + p^{*}$. Moreover, because $p^{*} > \mu$, the conjectured distribution must be such that $\int_{\ubar{w}}^{\bar{w}}d\tilde{F}_{X|P}\left(x|p^{*}\right) > 0$ for some $\ubar{w} \leq \bar{w} < p^{*}$. A firm's payoff from any belief $x \geq z^{*} + p^{*}$ is $p^{*}$ and from any belief $x < p^{*}$ is $0$, so it can deviate profitably by fusing the measure on $\left[\ubar{x}, \bar{x}\right]$ with a fraction, $\epsilon > 0$, of the measure on $\left[\ubar{w}, \bar{w}\right]$ (collapsing them to their barycenter). We conclude that  $z^{*} < 0$. 

Define \[F_{X|P}^{-}\left(a|p^{*}\right) \coloneqq \sup_{w < a}F_{X|P}\left(w|p^{*}\right) \text{ ,}\]
and observe that $F^{-}_{X|P}\left(p^{*}|p^{*}\right) < 1$, or else a firm would make $0$ profit from choosing $p^{*}$. Furthermore, any values $x < p^{*}$ yield a firm a profit of $0$ and since $p^{*} > \mu$, $F^{-}_{X|P}\left(p^{*}|p^{*}\right) > 0$. Consequently, a firm must have $F_{X|P}\left(p^{*}|p^{*}\right) = 1$ since otherwise it could (as in the previous paragraph) fuse a positive measure of values strictly below $p^{*}$ with a positive measure of values strictly above $p^{*}$ and obtain a higher payoff. 

Define \[q^{*} = q\left(p^{*}\right) \coloneqq F_{X|P}\left(p^{*}|p^{*}\right)  - F^{-}_{X|P}\left(p^{*}|p^{*}\right) \text{ ,}\] i.e., $q^{*}$ is the size of the mass point that $F_{X|P}\left(x|p^{*}\right)$ places on $p^{*}$. Accordingly, $z^{*} = -c/q^{*} < 0$. 

If $\mu \geq p^{*} > \mu - c$ is chosen on-path, then if the associated reservation value, $z^{*} \geq 0$, a firm must be selected with certainty, conditional on being visited. This is because $z^{*} \geq 0 > \mu - c - p^{*}$, which is the minimal reservation value that can be induced. Thus, if a firm were not selected with certainty, conditional on being visited, it could always fuse values strictly below and strictly above $z^{*} + p^{*}$. Accordingly, all values must be weakly greater than $p^{*}$, i.e., $F^{-}_{X|P}\left(p^{*}|p^{*}\right) = 0$. If the associated reservation value, $z^{*}$ is strictly negative, the result is trivial since a firm could always just provide no information (choose the degenerate distribution $\delta_{X}\left(\mu\right)$) and be selected with certainty, conditional on being visited.

Finally, if $\mu - c \geq p^{*} > 0$, then the minimum reservation value that can be induced is $\mu - c - p^{*} \geq 0$. Either $F^{-}_{X|P}\left(z^{*} + p^{*}|p^{*}\right) = 0$, in which case $z^{*} = \mu - c - p^{*}$ and so a firm is clearly selected for sure, conditional on being visited; or $F^{-}_{X|P}\left(z^{*} + p^{*}|p^{*}\right) > 0$ and so $z^{*} > \mu - c - p^{*}$ and $F_{X|P}\left(z^{*} + p^{*}|p^{*}\right) < 1$. In that case either a firm is selected for sure or it can deviate profitably by fusing portions of the measure strictly above and below $z^{*} + p^{*}$.
\end{proof}

Next, because there is no active search, we find that firms cannot make positive profits in any symmetric equilibrium. 

\begin{lemma}\label{nopos}
There are no symmetric equilibria in which firms make nonzero profits. Equivalently, there are no symmetric equilibria in which firms charge any price other than $0$.
\end{lemma}
\begin{proof}

Let $\hat{p}$ be the maximal price that is chosen on-path, with associated reservation value $\hat{z}$. 

First, suppose that $\hat{p} \leq \mu$. From Lemmata \ref{ppone} and \ref{aux}, following any on-path $p$, a firm must be chosen for sure, conditional on being visited. Clearly, $\hat{p}$ must be chosen with strictly positive probability on-path, since there is no active search. Otherwise, choosing that price would guarantee that that firm is never visited, yielding a profit of $0$.

Consequently, all values, $x$, must be such that $x - \hat{p} \geq \max\left\{0,\hat{z}\right\}$. Thus, $\hat{z} = \mu - c - \hat{p}$. But then a firm can deviate profitably to some price $\hat{p} - \eta$, where $\eta > 0$, and the degenerate distribution $\delta_{X}\left(\mu\right)$. Because $\mu - c - \hat{p} + \eta$ is the minimum reservation value such a deviation could induce, the deviating firm obtains a discrete jump up in its payoff. 

Second, suppose that $\hat{p} > \mu$. The probability that firms are choosing prices that are weakly greater than $\mu$ must be strictly positive (or else choosing $\hat{p}$ would result in $0$ profit for a firm). From a consumer's viewpoint, all on-path prices $p^{*} \geq \mu$ are equivalent. Indeed, for all such $p^{*} \geq \mu$, $F_{X|P}\left(p^{*}|p^{*}\right) = 1$. 

A firm's payoff from choosing any $p^{*} \geq \mu$, conditional on being visited, is $q\left(p^{*}\right)p^{*}$, where, recall, $q\left(p^{*}\right)$ is the size of the mass point on $p^{*}$. Moreover, by definition, $q\left(p^{*}\right)p^{*} \leq \mu$. But then, a firm can deviate profitably by choosing some price $p = \mu - \epsilon$, for a small but strictly positive $\epsilon$, and the degenerate distribution $\delta_{X}\left(\mu\right)$. No matter a consumer's belief about the firm's distribution after observing this price, a deviating firm must still be visited (and thus purchased from eventually) before any firm that is choosing $p^{*} \geq \mu$. Accordingly, for a sufficient small $\epsilon > 0$, a firm will receive a discrete jump in its payoff and so there exists a profitable deviation.

\end{proof}

\begin{lemma}\label{minres}
There exist no symmetric equilibria with active search in which firms make zero profits. 
\end{lemma}
\begin{proof}
Obviously, in any symmetric equilibrium in which firms make zero profits, then if a firm chooses any price $p > 0$ it must be visited with probability $0$. Thus, we impose that firms post prices $p = 0$. If $\mu - c > 0$, it suffices to to show that there are no equilibria in which firms choose distributions over values that induce reservation values strictly greater than $\mu - c$.

Suppose for the sake of contradiction that such an equilibrium exists. Then, a firm's distribution over values must be such that the probability of a realization that is weakly below $\mu - c - \gamma$ (with $\gamma > 0$) must occur with probability $\alpha > 0$. However, then a firm can deviate profitably to some price $p$ that satisfies $0 < p < \min\left\{\gamma, \mu-c\right\}$ and the degenerate distribution $\delta_{X}\left(\mu\right)$. Even should a consumer assign it the most pessimistic belief about its distribution, yielding a reservation value of $\mu - c- p$, the firm would still be visited (and selected) with strictly positive probability, yielding a positive profit. 

If $\mu -c \leq 0$, there is no active search on-path since any value a consumer obtains at the first firm is at least weakly greater than the reservation values of the other firms.

\end{proof}
At last, we arrive at the second main result.

\begin{theorem}\label{main2}
Any symmetric equilibrium must be one in which each firm posts price $p = 0$ and chooses a distribution that induces the minimum reservation value, $\mu - c$.
\end{theorem}
\begin{proof}
First, we establish that such purported equilibria are equilibria. Trivially, a firm cannot deviate profitably by choosing a different distribution over values (and keeping price $p = 0$). Next, we stipulate that a consumer assumes that any firm that deviates to a price $p \neq 0$ chooses a completely uninformative distribution. Accordingly, that firm will never be visited and therefore such a deviation is not profitable.

Second, uniqueness follows from Lemmata \ref{ppone}, \ref{nopos}, and \ref{minres}.
\end{proof}

In any symmetric equilibrium, each firm's profit is $0$ and each consumer's welfare is just $\mu$ (they do not incur a search cost, $c$, since that applies only to searches beyond the first firm). Although consumers obtain some surplus, there is no benefit from increased competition, i.e., a market with two firms is just as good for a consumer as one with infinitely many.

\section{Discussion}

\cite{asher} provides a compelling resolution for the Diamond paradox. Product differentiation and imperfect information allow for more realistic and empirically-relevant equilibria in which consumers search and in which competitive forces have real effects on the pricing decisions of firms. Indeed, a key aspect of Wolinsky's model is that consumers obtain information about their match values upon visiting firms. 

Crucially, that information \textit{must come from somewhere} and in particular, firms have a say in how much information a consumer's visit will glean about their respective products. This paper suggests a weakness in using imperfect information and horizontal differentiation alone to generate competitive pricing and active search in search models. Namely, if firms may choose how much information to provide about their products and that information is unadvertised then the Diamond paradox reemerges, despite the market's heterogeneity. Modifications to the model--like, e.g., assuming that a fraction of consumers can search without cost \'a la \cite{sta}--are needed to restore active search.

It is important to keep in mind that, just as the Diamond result requires that consumers learn prices only after paying a search cost, the hidden nature of the firms' information choices is essential to the results that we encounter here. If information itself can be advertised--i.e., consumers can observe each firm's chosen distribution over values without paying a visit cost--there are equilibria in which consumers search and obtain positive surplus. This situation with posted information is the subject of \cite{whau}. 

Furthermore, it seems quite plausible that although firms can advertise information policies \textit{partially}, they have no way of fully specifying or committing to more information than a certain (minimal amount). Consequently, the strong incentives for firms to deviate and under-provide information we find here suggest that firms will not provide more information than they can (explicitly) commit to. Casual observation suggests that this is somewhat accurate--there are many anecdotes of car dealerships reneging on advertised test-drives\footnote{See, e.g., some of the stories detailed in the following thread: \url{https://www.reddit.com/r/cars/comments/66r9ze/ever_had_a_dealer_not_let_you_test_drive_a_car/}. One user describes a test-drive consisting merely of a drive around the dealership's parking lot.} or real-estate agents whisking progressive tenants through their apartment tours.

\bibliography{sample.bib}

\begin{thebibliography}{17}
\providecommand{\natexlab}[1]{#1}
\providecommand{\url}[1]{\texttt{#1}}
\expandafter\ifx\csname urlstyle\endcsname\relax
  \providecommand{\doi}[1]{doi: #1}\else
  \providecommand{\doi}{doi: \begingroup \urlstyle{rm}\Url}\fi

\bibitem[Armstrong and Zhou(2011)]{az}
Mark Armstrong and Jidong Zhou.
\newblock Paying for prominence.
\newblock \emph{The Economic Journal}, 121\penalty0 (556):\penalty0 F368--F395,
  2011.

\bibitem[Armstrong and Zhou(2019)]{za}
Mark Armstrong and Jidong Zhou.
\newblock Consumer information and the limits to competition.
\newblock \emph{Mimeo}, November 2019.

\bibitem[Au and Whitmeyer(2020{\natexlab{a}})]{perc}
Pak~Hung Au and Mark Whitmeyer.
\newblock Attraction versus persuasion.
\newblock \emph{Mimeo}, November 2020{\natexlab{a}}.

\bibitem[Au and Whitmeyer(2020{\natexlab{b}})]{whau}
Pak~Hung Au and Mark Whitmeyer.
\newblock Advertised prices and information.
\newblock \emph{Mimeo}, November 2020{\natexlab{b}}.

\bibitem[Board and Lu(2018)]{board}
Simon Board and Jay Lu.
\newblock Competitive information disclosure in search markets.
\newblock \emph{Journal of Political Economy}, 126\penalty0 (5):\penalty0
  1965--2010, 2018.

\bibitem[Diamond(1971)]{diamond}
Peter~A Diamond.
\newblock A model of price adjustment.
\newblock \emph{Journal of Economic Theory}, 3\penalty0 (2):\penalty0 156 --
  168, 1971.

\bibitem[Dogan and Hu(2018)]{hu}
Mustafa Dogan and Ju~Hu.
\newblock Consumer search and optimal information.
\newblock \emph{Mimeo}, August 2018.

\bibitem[Ellison and Wolitzky(2012)]{wol}
Glenn Ellison and Alexander Wolitzky.
\newblock A search cost model of obfuscation.
\newblock \emph{The RAND Journal of Economics}, 43\penalty0 (3):\penalty0
  417--441, 2012.

\bibitem[Elton and Hill(1992)]{hill}
J.~Elton and T.~P. Hill.
\newblock Fusions of a probability distribution.
\newblock \emph{The Annals of Probability}, 20\penalty0 (1):\penalty0 421--454,
  1992.

\bibitem[Hwang et~al.(2018)Hwang, Kim, and Boleslavsky]{hkb}
Ilwoo Hwang, Kyungmin Kim, and Raphael Boleslavsky.
\newblock Competitive advertising and pricing.
\newblock \emph{Mimeo}, 2018.

\bibitem[Janssen and Shelegia(2015)]{janssen2015consumer}
Maarten Janssen and Sandro Shelegia.
\newblock Consumer search and double marginalization.
\newblock \emph{American Economic Review}, 105\penalty0 (6):\penalty0
  1683--1710, 2015.

\bibitem[Janssen and Shelegia(2020)]{janssen2020beliefs}
Maarten Janssen and Sandro Shelegia.
\newblock Beliefs and consumer search in a vertical industry.
\newblock \emph{Journal of the European Economic Association}, 18\penalty0
  (5):\penalty0 2359--2393, 2020.

\bibitem[Kamenica and Gentzkow(2011)]{kam}
Emir Kamenica and Matthew Gentzkow.
\newblock Bayesian persuasion.
\newblock \emph{The American Economic Review}, 101\penalty0 (6):\penalty0
  2590--2615, 2011.

\bibitem[Stahl(1989)]{sta}
Dale~O. Stahl.
\newblock Oligopolistic pricing with sequential consumer search.
\newblock \emph{The American Economic Review}, 79\penalty0 (4):\penalty0
  700--712, 1989.

\bibitem[Weitzman(1979)]{wei}
Martin~L. Weitzman.
\newblock Optimal search for the best alternative.
\newblock \emph{Econometrica}, 47\penalty0 (3):\penalty0 641--654, 1979.

\bibitem[Wolinsky(1986)]{asher}
Asher Wolinsky.
\newblock {True Monopolistic Competition as a Result of Imperfect
  Information*}.
\newblock \emph{The Quarterly Journal of Economics}, 101\penalty0 (3):\penalty0
  493--511, 08 1986.

\bibitem[Zhou(2020)]{zhou}
Jidong Zhou.
\newblock Improved information in search markets.
\newblock \emph{Mimeo}, October 2020.

\end{thebibliography}

\end{document}